\documentclass[]{article}
\usepackage{graphicx}

\usepackage[hidelinks]{hyperref}

\usepackage{authblk}

\usepackage{glossaries}
\usepackage{booktabs}
\newacronym{acr:pcsp}{PCS}{Path-based Coflow Scheduling}
\newacronym{acr:cosp}{COS}{Concurrent Open Shop}
\newacronym{acr:bcsp}{BCS}{Bipartite Coflow Scheduling}
\glsdisablehyper

\usepackage{subcaption}
\usepackage{tikz}

\usepackage{amsmath}
\usepackage{amssymb}
\usepackage{dsfont}	
\usepackage{amsthm}
\newtheorem{definition}{Definition}
\newtheorem{theorem}{Theorem}
\newtheorem{lemma}{Lemma}

\newcommand{\st}{\text{s.t.}}
\newcommand{\eps}{\varepsilon}
\newcommand{\cik}{c_i^{(k)}}
\newcommand{\cjk}{c_j^{(k)}}
\newcommand{\Lk}{L^{(k)}}
\newcommand{\Lik}{L_i^{(k)}}
\newcommand{\Ljk}{L_j^{(k)}}
\newcommand{\Lil}{L_i^{(l)}}
\newcommand{\cijk}{c_{i,j}^{(k)}}

\newcommand{\fjk}{f_j^{(k)}}
\newcommand{\pjk}{P_j^{(k)}}
\DeclareMathOperator*{\opt}{opt}
\DeclareMathOperator*{\avg}{avg}

\usepackage[ruled,linesnumbered]{algorithm2e}
\SetKwInOut{Input}{Input}
\SetKwInOut{Output}{Output}

\begin{document}
\title{Minimization of Weighted Completion Times in Path-based Coflow Scheduling\thanks{Work supported by Deutsche Forschungsgemeinschaft (DFG), GRK 2201 and by the European Research Council, Grant Agreement No.\ 691672.}}
\author[1,2]{\small Alexander Eckl}
\author[1]{Luisa Peter}
\author[3]{Maximilian Schiffer}
\author[1]{Susanne Albers}
\affil[1]{Department of Informatics\\ Technical University of Munich\\ Boltzmannstr. 3\\ 85748 Garching, Germany\\\texttt{alexander.eckl@tum.de, luisa.peter@in.tum.de, albers@in.tum.de}}
\affil[2]{Advanced Optimization in a Networked Economy (AdONE)\\ Technical University of Munich\\ Arcisstra\ss e 21\\ 80333 Munich\\ Germany}
\affil[3]{TUM School of Management\\ Technical University of Munich\\ Arcisstra\ss e 21\\ \newline 80333 Munich\\ Germany\\
\texttt{schiffer@tum.de}}
\date{}

\maketitle
\begin{abstract}
\noindent
Coflow scheduling models communication requests in parallel computing frameworks where multiple data flows between shared resources need to be completed before computation can continue. 
In this paper, we introduce Path-based Coflow Scheduling, a generalized problem variant that considers coflows as collections of flows along fixed paths on general network topologies with node capacity restrictions. For this problem, we minimize the coflows' total weighted completion time. We show that flows on paths in the original network can be interpreted as hyperedges in a hypergraph and transform the path-based scheduling problem into an edge scheduling problem on this hypergraph. 

We present a $(2\lambda +1)$-approximation algorithm when node capacities are set to one, where $\lambda$ is the maximum number of nodes in a path. For the special case of simultaneous release times for all flows, our result improves to a $(2\lambda)$-approximation. Furthermore, we generalize the result to arbitrary node constraints and obtain a $(2\lambda \Delta +1)$- and a $(2\lambda \Delta)$-approximation in the case of general and zero release times, where $\Delta$ captures the capacity disparity between nodes. \\

\end{abstract}

\section{Introduction}

Parallel computing frameworks, e.g., MapReduce~\cite{DeanGhemawat2008}, Spark~\cite{ZahariaChowdhuryEtAl2010}, or Google Data{-}flow~\cite{GoogleDataflow2019}, are a central element of today's IT architecture and services, especially in the course of Big Data. Computing jobs in such a setting consist of a sequence of tasks, executed at different stages, such that in between data must be transferred from multiple output to multiple input ports. Herein, a consecutive task cannot start before all data is transferred to the next stage. Based on the \emph{coflow networking abstraction} (see \cite{ChowdhuryStoica2012}), we abstract such a data transfer as a set of data flows on a communication stage and refer to it as a coflow if all flows must be finished to allow for the execution of the next task of a computing job. This data transferring between computations can contribute more than 50\% to a jobs completion time (see \cite{ShiZhangEtAl2018}) such that minimizing coflow completion times remains a central challenge that heavily affects the efficiency of such environments.

So far, research on coflow scheduling has mainly focused on bipartite networks  \cite{AhmadiKhullerEtAl2017,ChowdhuryStoica2015,ChowdhuryZhongEtAl2014,QiuSteinEtAl2015}. 
Here, machines are uniquely divided into input and output ports and data can be transferred instantaneously via a direct link between each pair of input and output ports (see Figure~\ref{pic:bipartite_coflow}). Recently, research has shifted to more general network topologies. Jahanjou et al.\ \cite{JahanjouKantorEtAl2017} first introduced a variation of Coflow Scheduling where the underlying networks of machines is an arbitrary graph. Since then, this generalized problem has been considered more extensively~\cite{ChowdhuryKhullerEtAl2019,ShiZhangEtAl2018}. Applications arise for grid computing projects, i.e., inter-data center communication, where parallel computing tasks are executed on multiple but decentralized high-power computing units \cite{ChowdhuryKhullerEtAl2019}. A node in the underlying network may represent a machine, a data center, or an intermediate routing point. While these recent works indeed generalize the case of bipartite coflow scheduling, a fully generalizable approach is still missing, since infinite router/machine capacities during communication remain a central assumption in all current modeling approaches.

\begin{figure}[bt]
\centering
	\begin{subfigure}[b]{.45\textwidth}
		\centering
		\begin{tikzpicture}[main node/.style={draw,circle,inner sep=2pt,fill,thick},>=latex]
		\foreach \v in {1,2,3}{
			\node[main node] (I\v) at (0,1.5*\v-0.75) {};
		}
		\foreach \v in {1,2,3,4}{
			\node[main node] (O\v) at (3,1.5*\v-1.5) {};
		}
		\foreach \i in {1,2,3}{
			\foreach \o in {1,2,3,4}{
				\draw (I\i) -- (O\o);
			}
		}
		\node (I) at (0,5) {\footnotesize input ports};
		\node (O) at (3,5) {\footnotesize output ports};
		\draw[ultra thick,->,green!50!black] (I1) -- (O2);
		\draw[ultra thick,->,green!50!blue] (I1) -- (O4);
		\draw[ultra thick,->,green] (I2) -- (O2);
		\draw[ultra thick,->,green!50!yellow] (I3) -- (O1);
		\end{tikzpicture}
		\caption{A conventional coflow on a bipartite network}
		\label{pic:bipartite_coflow}
	\end{subfigure}\quad
	\begin{subfigure}[b]{.45\textwidth}
		\centering
		\begin{tikzpicture}[main node/.style={draw,circle,inner sep=2pt,fill,thick},>=latex]
		\node[main node] (1) at (0.5,2) {};
		\node[main node] (2) at (2,0) {};
		\node[main node] (3) at (2,4) {};
		\node[main node] (4) at (3,2) {};
		\node[main node] (5) at (4,4) {};
		\node[main node] (6) at (5,2) {};
		\draw (1) -- (2);
		\draw (1) -- (3);
		\draw (3) -- (4);
		\draw (3) -- (5);
		\draw (4) -- (5);
		\draw (4) -- (6);
		\draw (5) -- (6);
		\draw[cap=round,rounded corners,ultra thick,->,red!50!white] ([xshift=-1.5pt,yshift=-3pt]5.west) --  ([xshift=0pt,yshift=3.5pt]4.north);
		\draw[cap=round,rounded corners,ultra thick,->,red!60!blue] ([xshift=-0.75pt,yshift=2pt]4.west) -- ([yshift=-1pt]3.south) -- ([xshift=0.75pt]1.east) -- ([yshift=1pt]2.north);
		\draw[cap=round,rounded corners,ultra thick,->,red!70!black] ([yshift=1.5pt]1.north) -- ([yshift=1pt]3.north) -- ([xshift=0.75pt,yshift=1pt]5.north) -- ([xshift=1pt,yshift=1pt]6.north);
		\draw[cap=round,rounded corners,ultra thick,->,red] ([xshift=1pt,yshift=-2pt]3.east) -- ([xshift=1pt,yshift=1pt]4.north) -- ([xshift=-2.5pt,yshift=0.5pt]6.north);
		\end{tikzpicture}
		\caption{A path-based coflow on an arbitrary network}
		\label{pic:path-based_coflow}
	\end{subfigure}
	\caption{Illustration of the prevalent bipartite coflow setting as opposed to the more general concept of path-based coflows}
	\label{pic:bipartite_vs_path-based}
\end{figure}
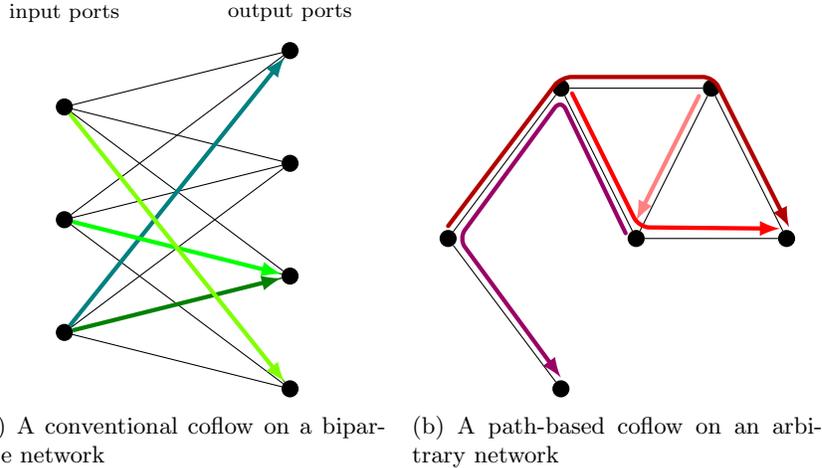


To close this gap in the existing literature, we introduce the concept of \emph{\gls{acr:pcsp}}, which considers coflow scheduling in the more general setting of Jahanjou et al., where coflows consist of multiple data flows that may run between any two machines of the underlying network on a fixed path of finite length (see Figure~\ref{pic:path-based_coflow}). Additionally, we impose that machines can handle only a single flow type at any time. We further generalize the problem to non-uniform node capacities to consider different router capacities.

In the following, we first give a formal definition of \gls{acr:pcsp}, before we review related literature and detail the contribution of our work.

\subsection{Definition of the Path-based Coflow Scheduling Problem}
\label{subsec:def_path_based_coflow_scheduling}

Let $G_I=(V_I,E_I)$ be a multigraph with $m$ nodes. Every node corresponds to a machine and every edge to a communication line between two machines. A \emph{coflow} $k \in [n]$ with weight $w_k$ is a collection of \emph{flows} $\fjk$, $j \in [n_k]$, each sending $\cjk \in \mathbb{N}$ units of data along a given path $\pjk$ in the graph. For the longest flow-carrying path in $G$, we denote its number of nodes as $\lambda = \max_{k, j} \left| \pjk \right|$. Along all paths, we assume data transfer to be instantaneous.

For a given discrete and finite but sufficiently large time horizon with $t = 1, 2, \dots, T$ time steps, a \emph{schedule} assigns the execution of every flow of each coflow to $\cjk$ time steps, such that each node handles at most one unit of data per time step. To this end, a coflow and its flows have a \emph{release time} $r_k$ such that flows of coflow $k$ can only be scheduled from time step $t = r_k + 1$ onward. Each coflow has a \emph{completion time} $C_k$, which is the earliest time at which all flows related to $k$ have been executed.

In this setting, the objective of \gls{acr:pcsp} is to find a schedule that minimizes the weighted sum of completion times
\begin{equation}\label{eq:obj_func}
\min \ \sum_{k = 1}^n w_k C_k.
\end{equation}

\subsection{Related Work}
\label{subsec:related_work}

One may view our approach as coflow scheduling with underlying matching constraints on general network topologies. Accordingly, \gls{acr:pcsp} is related to different variants of Coflow Scheduling and to the \gls{acr:cosp} problem in general. In the following, we concisely review related work in these fields.
 
Within the emerging field of coflow scheduling, primarily \gls{acr:bcsp} (see Figure~\ref{pic:bipartite_coflow}) has been studied \cite{AhmadiKhullerEtAl2017,ChowdhuryStoica2015,ChowdhuryZhongEtAl2014,QiuSteinEtAl2015}. Ahmadi et al.\ presented the current state of the art, providing a 5-approximation for \gls{acr:bcsp} with release times, and a 4-approximation if release times are zero~\cite{AhmadiKhullerEtAl2017}. Recently, Shafiee and Ghaderi \cite{ShafieeGhaderi2018} achieved the same ratios based on a different LP formulation. 

Table~\ref{tab:differences_coflow_types} summarizes coflow variants on arbitrary graphs, which have not been extensively studied so far and can be divided into \emph{path-based}, \emph{circuit-based}, and \emph{packet-based} coflow scheduling: Jahanjou et al.\ focused on circuit-based coflows where a flow is a connection request between two nodes, every edge has limited capacity, but different jobs may send flow over the same edge \cite{JahanjouKantorEtAl2017}. They provide a $17.6$-approximation for circuit-based coflows with fixed paths. Recently, Chowdhury et al.\ improved this ratio to a randomized 2-approximation~\cite{ChowdhuryKhullerEtAl2019}. Jahanjou et al.\ also considered \emph{packet-based coflows} to be a set of packet transmissions from a source to a destination on a given network. Every edge can only serve at most one packet per time. Contrary to the circuit setting, packets are not transfered along the entire path instantaneously. They gave an $O(1)$-approximation for packet-based coflows with given paths.

\gls{acr:pcsp} has not been addressed so far and differs from circuit- and packet-based coflows as it allows only single unit-sized packets to be sent over each node. In contrast, previous approaches allow fractional data transmissions on nodes and links.

In all of the previously mentioned results, the path a flow takes through the network is assumed to be fixed. Several publications additionally introduce different methods of \emph{routing} for the flows in the network to further improve completion times, including \cite{ChowdhuryKhullerEtAl2019,JahanjouKantorEtAl2017,ShiZhangEtAl2018,ZhaoChenEtAl2015}. In this paper, we always assume paths are given by the problem definition and no routing takes place.

So far, providing an algorithm for \gls{acr:bcsp} with an approximation ratio less than $5$ resp.\ $4$ has not been successful. Im et al.\ \cite{ImMoseley2019} recently presented a $2$-approximation for Matroid Coflow Scheduling, where the family of flow sets that can be scheduled in a given time slot form a matroid. Since the flows that can be scheduled in a bipartite network do not form a matroid, this result does not improve the afore-mentioned ratios. A $2$-approximation for \gls{acr:bcsp} was claimed in \cite{ImPurohit2017}, which was then subsequently retracted.

Coflow Scheduling generalizes \gls{acr:cosp} \cite{LeungLiEtAl2007,MastrolilliQueyranneEtAl2010,WangCheng2007}, where we are given a set of machines $i \in [m]$, jobs $k \in [n]$, and every job $k$ has $\cik$ operations on machine $i$. The weight, release time, and completion time of a job are defined as for coflows. The goal is to minimize the weighted sum of completion times $\sum_{k = 1}^n w_k C_k$. Sachdeva and Saket showed in \cite{SachdevaSaket2013} that \gls{acr:cosp} is hard to approximate within a factor of $2 - \eps$ if $\text{P}\neq \text{NP}$ and therefore the same result holds for any variant of Coflow Scheduling. \gls{acr:cosp} admits a 3-approximation and 2-approximation for general and zero release times, respectively \cite{LeungLiEtAl2007}. During our algorithm, we will compute a solution to the underlying \gls{acr:cosp} inherent to any coflow instance. 

\begin{table}[btp]
	\centering
	\caption{Overview of the differences between coflow settings on arbitrary graphs} \label{tab:differences_coflow_types}
	\resizebox{\columnwidth}{!}{ 
		\begin{tabular}{llll}
			\toprule
			& & \textbf{constraints} & \textbf{data transfer along path} \\
			\midrule
			\emph{path-based} &  & (unit) node capacities & instantaneous, restricted to unit-sized packets\\
			\emph{circuit-based} & \cite{JahanjouKantorEtAl2017,ChowdhuryKhullerEtAl2019} & edge capacities & instantaneous, multi-transfer\\
			\emph{packet-based} & \cite{JahanjouKantorEtAl2017} & unit edge capacities & stepwise, multi-transfer\\
			\bottomrule
		\end{tabular}
	}
\end{table}

\subsection{Contribution}
\label{subsec:contribution}

With this paper, we are the first to introduce \gls{acr:pcsp}, which generalizes the well known \gls{acr:bcsp}. We present an approximation algorithm based on a novel edge scheduling problem on a hypergraph. Specifically, we show that flows can be interpreted as hyperedges instead of paths in the network, because they occupy all machines on their path simultaneously. Based on the same insight, we show that \gls{acr:pcsp} is at least as hard as vertex coloring and hence admits no constant approximation ratio.  Theorem~\ref{thm:result_det_path_based} states our main result with $\lambda$ being the number of vertices in a longest flow-bearing path of an instance graph.

\begin{theorem}
	\label{thm:result_det_path_based}
	There exists a $(2\lambda + 1)$-approximation for \gls{acr:pcsp} with arbitrary release times and a $2\lambda$-approximation for \gls{acr:pcsp} with zero release times.
\end{theorem}

Since no constant ratio can be achieved, some parameter of \gls{acr:pcsp} has to be assumed constant to make the problem tractable. We consider a natural restriction on $\lambda$, the maximum length of any occurring path. For small values of $\lambda$, our instance reduces to well-studied concurrent scheduling problems. In these cases, our ratios match the current state-of-the-art results. We refer to Section \ref{subsec:hardness_reduction} for details on the hardness of \gls{acr:pcsp} and the reduction from vertex coloring.

Section \ref{sec:methodology} details the proof of Theorem~\ref{thm:result_det_path_based}. It is possible to extend the special case where all release times are zero to arbitrary release times that are smaller than $\lambda$. Furthermore, the approximation guarantee can be slightly improved to a factor strictly smaller than $(2\lambda +1)$ and $2\lambda$, respectively. See Appendix \ref{sec:improvements} for details on how to obtain these minor improvements. Additionally, we generalize the algorithm to the case of non-uniform node capacities in Section~\ref{subsec:vertexConstraints}. We also show that it matches or improves the state of the art for several problem variants. First, for $\lambda \le 9$ we improve the deterministic state of the art for circuit-based coflows with unit capacities, which is a 17.6-approximation developed by Jahanjou et al. \cite{JahanjouKantorEtAl2017}. Refer to Appendix \ref{sec:equ_edge_node_constr} for more details on this statement. Second, for $\lambda=2$ our algorithm matches the state of the art for \gls{acr:bcsp}, a 5-approximation with release times and a 4-approximation without release times~\cite{AhmadiKhullerEtAl2017,ShafieeGhaderi2018}. Moreover, our algorithm yields the same ratios without the bipartiteness condition. Refer to Section \ref{subsec:bipartiteCS} for a detailed comparison. Third, with $\lambda = 1$ \gls{acr:pcsp} reduces to \gls{acr:cosp}. In this case, our algorithm matches the state of the art, yielding a 3-approximation with release times and a 2-approximation without release times. Overall, our approach seems to capture the difficulty of open shop scheduling with matching constraints well, especially if the parameter $\lambda$ is small.

\section{Methodology}\label{sec:methodology}

This section details the methodological foundation for Theorem~\ref{thm:result_det_path_based} and discusses the general hardness of our Problem.

We prove the theorem in two steps. First, we introduce an LP relaxation of \gls{acr:pcsp} in Section~\ref{sec:finding_deadlines}. Specifically, we reduce an instance of \gls{acr:pcsp} to an instance of \gls{acr:cosp} by ignoring matching constraints and considering each node individually. We derive deadlines for the coflows from the LP solution. These tentative deadlines lie provably close to the optimal solution of the LP. 

An important insight of this paper is that since flows occupy \emph{all} machines on their path simultaneously they can be interpreted as hyperedges instead of paths in the network. 
Thus, we transform every flow-path of the underlying graph into a hyperedge in Section~\ref{sec:edge_scheduling}. Here, we determine a schedule such that every edge still finishes within a factor of the previously found deadlines but no hyperedges that contain the same node overlap. We introduce a new problem called \textit{Edge Scheduling}, based on a hypergraph $G = (V,E)$ with release time $r_e$, and deadline $D_e$ for every edge $e \in E$. At each discrete time step $t = 1, 2, \dots, T$, we can schedule a subset of edges if they form a matching. The goal is to find, if possible, a feasible solution that schedules all edges between their release time and deadline.

In summary, we prove Theorem~\ref{thm:result_det_path_based} in Section~\ref{sec:edge_scheduling} based on the following rationale: The solution of the Edge Scheduling problem lies within a guaranteed factor of the deadlines constructed in step one. Since these deadlines were defined by the LP solution, which in turn is bounded by the optimal solution of the Coflow instance, the combined algorithm ultimately yields a provably good approximation factor.

Finally, the rest of Section \ref{sec:methodology} is dedicated to showing the hardness of \gls{acr:pcsp} via a polynomial reduction from Vertex Coloring.

In the remainder, we refer to coflows as \emph{jobs} and to flows as \emph{operations} to avoid ambiguous wording.

\subsection{Finding Deadlines with Good Properties}
\label{sec:finding_deadlines}

Let $I$ be an instance of \gls{acr:pcsp} with its underlying graph $G_I = (V_I,E_I)$. We introduce variables $C_k$  to denote the completion time of each job $k$. Further, we define the \emph{load} of job $k$ on machine $i$ as the sum of all operations of $k$ that go through node $i$:
\begin{equation*}
\Lik= \sum_{j:\ i \in \pjk} \cjk.
\end{equation*}
For any subset $S \subset [n]$ and any machine $i \in [m]$, we define the variables $f_i(S)$:
\begin{equation*}
f_i(S) = \frac{1}{2} \cdot \left(\sum_{k\in S} \Big( \Lik \Big)^2 + \Big(\sum_{k\in S} \Lik \Big)^2 \right).
\end{equation*}

With this notation our LP relaxation results to:
\begin{align}
\min	&&\sum_{k = 1}^n w_k C_k 	&				&		& \label{eq:LP_obj} \\
\st		&&C_k						&\ge r_k + \Lik	&\quad	& \forall k \in [n], \forall i \in [m] \label{const:single_lower_bound} \\
		&&\sum_{k\in S}	\Lik C_k	&\ge f_i(S)		&		& \forall S \subset [n], \forall i \in [m]. \label{const:subset_lower_bound}
\end{align}
The first set of constraints \eqref{const:single_lower_bound} obtains a lower bound for the completion time of a single job $k$ based on its release time. 
The second set of constraints \eqref{const:subset_lower_bound} provides a lower bound on the completion time of any set of jobs $S$. 
Note that \eqref{const:subset_lower_bound} has been used frequently in \gls{acr:cosp} and \gls{acr:bcsp} \cite{Queyranne1993,LeungLiEtAl2007,GargKumarEtAl2007,AhmadiKhullerEtAl2017} and, although the number of constraints is exponential, can be polynomially separated \cite{Queyranne1993}. Accordingly, we can solve \eqref{eq:LP_obj}--\eqref{const:subset_lower_bound} using the ellipsoid method~\cite{GroetschelLovaszEtAl1993}.

We denote by $(C_k^*)_{k \in [n]}$ an optimal solution to the LP and consider (w.l.o.g.) the jobs to be ordered s.t.\ $C_1^* \le \hdots \le C_n^*$.

\begin{lemma}[{\cite[Lemma 11]{LeungLiEtAl2007}}]
\label{lem:bound_on_LP_sol}
	For all jobs $k \in [n]$ and all machines $i \in [m]$ the following holds:
	\begin{equation*}
		C_k^* \ge \frac{1}{2} \sum_{l=1}^k \Lil.
	\end{equation*}
\end{lemma}

\begin{proof}
	Let $k \in [n], i \in [m], S = \{1,\dots,k\}$. Since $C_1^*,\dots,C_n^*$ is a feasible solution of the LP, it must fulfill \eqref{const:subset_lower_bound}, such that
	\begin{align*}
		\sum_{l=1}^k \Lil C_l^* &= \sum_{l\in S} \Lil C_l^* \\
								&\ge f_i(S) \\
								&\ge \frac{1}{2} \cdot \left( \sum_{l=1}^k \Lil \right)^2.
	\end{align*}
	Accordingly, we estimate the completion time of job $k$ as follows:
	\begin{align*}
		C_k^* 	&\ge \frac{\sum_{l=1}^k \Lil C_k^*}{\sum_{l=1}^k \Lil} \\
				&\ge \frac{\sum_{l=1}^k \Lil C_l^*}{\sum_{l=1}^k \Lil} \\
				&\ge \frac{1}{2} \sum_{l=1}^k \Lil.
	\end{align*}
\end{proof}

We now define a deadline $D_k$ for every job $k$. We utilize $D_k$ in Section~\ref{sec:edge_scheduling} to define a partial order on the operations of the instance.
With Lemma~\ref{lem:bound_on_LP_sol}, we estimate $D_k$ for all $k \in [n]$:
\begin{equation}
\label{eq:lower_bound_deadlines}
	D_k := 2 \cdot C_k^* \ge \sum_{l=1}^k \Lil.
\end{equation}

\subsection{The Edge Scheduling Algorithm}
\label{sec:edge_scheduling}

In this section, we design our edge scheduling algorithm. First, based on the deadlines $D_k$, we define a partial order on the operations of $I$. For every operation~$j$, it induces an upper bound on the number of preceding operations that share a node with $j$. With this order, we can then devise our edge scheduling algorithm.

\paragraph{Operation Order Based on Deadlines.}
We transform $G_I$ into a hypergraph $G = (V,E)$. While the node set remains the same ($V = V_I$), we derive the hyperedges from the operations of the instance, i.e., the edge set $E$ consists of all hyperedges constructed in the following way: Let $\fjk$ be an operation on a path $\pjk$. Then, we add \emph{for each of the $c_j^{(k)}$ units of data sent by $\fjk$} a corresponding hyperedge ${e := \{v \in V: v \in \pjk\}}$, such that it consists of all nodes of the operation's path. Note that we have $|e| \le \lambda$ for all $e \in E$ with the maximum path-length $\lambda$. By so doing, we receive $c_j^{(k)}$ identical edges for every operation. Furthermore, let $k_e \in [n]$ denote the job corresponding to the operation of edge $e$. We set the release time $r_e := r_{k_e}$ and the deadline $D_e := D_{k_e}$ of $e$. 

Note that adding a hyperedge for every unit of data sent over an operation may not be polynomial in the length of the input. We can avoid a pseudo-polynomial construction of the hypergraph if we only add one hyperedge $e$ for each operation $\fjk$ and assign a value of $c_j^{(k)}$ to $e$. As it makes no difference to the algorithm whether one edge is scheduled $c_j^{(k)}$ times or $c_j^{(k)}$ identical edges are each scheduled once, we choose the latter variant, which allows for a clearer analysis.

We now consider the line graph $L=L(G)$ of $G$. Note that $L$ is always a simple graph, although $G$ is a hypergraph with possibly multiple edges. Let $e$ and $f$ be hyperedges of $G$ with a common vertex $v$. We then say the edge $\{e,f\} \in E(L)$ \emph{originated} from $v$.

As a basis for our algorithm, we define an order on the operations, i.e., the hyperedges of $G$ or the vertices of $L$, using the notion of \emph{orientations} and \emph{kernels}.

\begin{definition}
	Let $G$ be a graph. An \emph{orientation} $O$ of $G$ is a directed graph on the same vertex and edge set, where every edge in $G$ is assigned a unique direction in $O$. We define $d^+(v)$ as the set of outgoing edges at a vertex $v \in V(O)$.
\end{definition}
\begin{definition}
\label{def:kernel}
	Let $G$ be a graph and let $O$ be an orientation of $G$. We call an independent set ${U \subset V(O)}$ a \emph{kernel} of $O$, if for all $v \notin U$ there is an arc directed from $v$ to some vertex in $U$.
\end{definition}

W.l.o.g.\ we order the vertices of $L$ (i.e., the hyperedges of $G$) by the converted deadlines obtained from the job deadlines of Section \ref{sec:finding_deadlines}. For vertices that have the same deadline, we use an arbitrary order. Let this order be such that $D_{e_1} \le \hdots \le D_{e_{|V(L)|}}$, which is consistent with the ordering $D_{k_{e_1}} \le \hdots \le D_{k_{e_{|V(L)|}}}$ obtained from the deadlines of Section \ref{sec:finding_deadlines}. Let $N(e)$ be the set of neighbours of an edge $e$ in $L$. We construct an orientation $O$ of $L$ with Algorithm~\ref{alg:orientation}.

\begin{algorithm}[t]
	$V(O)\leftarrow V(L)$\;
	\For{$j \leftarrow |V(L)|$ \KwTo $1$}{
		\ForEach{$e \in N(e_j)$}{
			\If{$\{e_j,e\} \in E(L)$ not oriented yet}{
				add arc $(e_j,e)$ to $O$\;
			}		
		}
	}
	\caption{Orientation of the line graph} \label{alg:orientation}
\end{algorithm}

\noindent The algorithm simply directs any edge of $L$ such that the endpoint with the higher deadline points to the one with the lower deadline. Specifically, $O$ shows the characteristics described in Lemma~\ref{lem:orient_properties}.

\begin{lemma}
\label{lem:orient_properties}
	An orientation $O$ constructed by Algorithm \ref{alg:orientation} has the following properties:
	\begin{enumerate}
		\item Any vertex $e \in V(L)$ satisfies the inequality $|d^+(e)| \le \lambda (D_e -1)$.
		\item $O$ is kernel-perfect, i.e., every induced subgraph of $O$ has a kernel.
	\end{enumerate}
\end{lemma}

\begin{proof}
	We prove Lemma \ref{lem:orient_properties} in two steps.
	\begin{enumerate}
		\item Consider an arbitrary vertex of $L$ representing edge $e$ with $j$ being the index of $e$ in the ordering of the edges $D_{e_1} \le \hdots \le D_{e_{|V(L)|}}$. Recall that by Algorithm~\ref{alg:orientation}, $e$ has only outgoing arcs in $L$ to vertices in the set $\{e_1,\dots,e_{j-1}\}$.
	
		In $G$, $e$ is a hyperedge with at most $\lambda$ endpoints. Let $v$ be an endpoint of $e$ and let $d_v^+(e) \subset d^+(e)$ be the set of outgoing arcs from $e$ that originated from $v$ during the construction of the line graph.
	
		We now focus on the cardinality of $d_v^+(e)$: the endpoint of any arc from this set must lie in $\{e_1,\dots,e_{j-1}\}$. Recall that by $k_e$ we denote the job corresponding to an edge $e$. For all edges $f \in \{e_1,\dots,e_{j-1}\}$ we have $D_{k_f} \le D_{k_e}$. Hence, the same holds for all edges $f$ that are the endpoint of an arc in $d_v^+(e)$. Therefore, we obtain
		\begin{align}
		\left|d_v^+(e)\right| 	&\le \left|\{f \in E\setminus\{e\}: f \text{ contains } v \text{ and } D_{k_f} \le D_{k_e} \}\right| \nonumber\\
								&= \left|\{f \in E: f \text{ contains } v \text{ and } D_{k_f} \le D_{k_e} \}\right| - 1 \nonumber\\
								&= \sum_{l=1}^{k_e} L_v^{(l)} - 1 \label{eq:proof_cardinality}\\
								&\le D_{k_e} - 1. \label{eq:proof_cardinality2}
		\end{align}
		To derive \eqref{eq:proof_cardinality}, we observe that the load on machine $v$ up to job $k_e$ is equal to the number of edges containing $v$ from jobs with a smaller or equal deadline. The final step \eqref{eq:proof_cardinality2} results from \eqref{eq:lower_bound_deadlines}.
	
		Since $e$ has at most $\lambda$ endpoints in $G$, we conclude
		\begin{equation*}
		\left|d^+(e)\right| \le \sum_{v \in e} \left|d_v^+(e)\right| \le \lambda \cdot (D_{k_e} - 1) = \lambda \cdot (D_e - 1).
		\end{equation*}
	
		\item We note that any digraph without directed cycles of odd length is kernel-perfect {\cite{Richardson1946}}. Additionally, we observe that $O$ does not contain any directed cycles to begin with.
	\end{enumerate}
\end{proof}

\paragraph{Edge Scheduling.}
With these preliminaries, we devise our Edge Scheduling algorithm as described in Algorithm~\ref{alg:edge_scheduling}. This algorithm finds a feasible edge schedule on $G$ such that no edge is scheduled later than $r_e + \lambda D_e$ (see Lemma~\ref{lem:edge_sched_solution}). 

\begin{lemma}
\label{lem:edge_sched_solution}
	Algorithm \ref{alg:edge_scheduling} finds a feasible solution for Edge Scheduling on a given hypergraph $G$, s.t.\ every hyperedge $e$ is scheduled not later than $r_e + \lambda D_e$.
\end{lemma}

\begin{proof}
	We note that any induced subgraph of $O$ has a kernel (see Lemma~\ref{lem:orient_properties}). Hence, we can find a kernel $U$ in each iteration of the algorithm because the modified graph remains an induced subgraph of the original orientation. Refer to Appendix \ref{sec:constructing_kernel} on how to construct a kernel in a cycle-free directed graph. Accordingly, Algorithm~\ref{alg:edge_scheduling} is well defined.
	
\begin{algorithm}[!b]
	\Input{A hypergraph $G=(V,E)$, an orientation $O$ of $L(G)$, a release time $r_e$ and a deadline $D_e$ for every hyperedge $e \in E\ \st\ d^+(e) \le \lambda (D_e - 1)$.}
	\Output{A feasible edge schedule on $G$.}
	$T \leftarrow \max_{e \in E} r_e + \lambda \cdot \max_{e \in E} D_e$\;
	\For{$t \leftarrow 1$ \KwTo $T$}{
				$O' \leftarrow$ the induced subgraph of $O$ on all vertices $e$ with $t > r_e$\;
				$U \leftarrow$ a kernel of $O'$\;
				\ForEach{$e \in U$}{
					schedule $e$ in timeslot $t$\;				
				}
				remove $U$ from $O$\;
	}
\caption{Edge Scheduling}\label{alg:edge_scheduling}
\end{algorithm}
	 
	For an arbitrary hyperedge $e$ of $G$, assume that in any iteration of the algorithm we have $d^+(e) = \emptyset$ and $e$ is already released. Then, $e$ is scheduled at the current time slot because $e$ lies in the kernel $U$ as it has no outgoing edges and $e\in O'$. Hence, it suffices to prove that for any hyperedge $e$ of $G$ after at most $r_e + \lambda D_e - 1$ iterations $d^+(e) = \emptyset$ holds.
	 
	The orientation $O$ fulfills $|d^+(e)| \le \lambda (D_e - 1)$ in the beginning of the algorithm (see Lemma~\ref{lem:orient_properties}). We note that for any iteration in which $t \le r_e$, hyperedge $e$ is not considered to be scheduled at all, which is necessary to satisfy the release time constraint. We now consider all iterations $r_e +1, r_e +2,\dots , T$. In each of these iterations, $e\in O'$ holds because $t > r_e$ and two cases remain:
	\begin{enumerate}
	\item If $e \in U$ at any point before iteration $r_e + \lambda D_e$, the result is immediate.
	\item If, on the other hand, $e \notin U$, then $e$ must have an outgoing edge to some $e' \in U$ by the kernel property of $U$. As $e'$ gets removed from $O$ at the end of the iteration, $e$ loses at least one outgoing edge. Hence, after at most $\lambda (D_e - 1) \le \lambda D_e - 1$ such iterations, we have $|d^+(e)| = 0$. 
	\end{enumerate}
	This concludes the proof.
\end{proof}

Given this upper bound on the scheduled time for every edge, we prove Theorem \ref{thm:result_det_path_based}.

\begin{proof}[Proof of Theorem \ref{thm:result_det_path_based}]
	We consider a given instance $I$ of \gls{acr:pcsp}. Then, we can solve the LP relaxation \eqref{eq:LP_obj} to receive a set of solutions $C_k^*$ for all jobs $k$ (see Section~\ref{sec:finding_deadlines}). We define deadlines $D_k = 2 \cdot C_k^*$. Note that we have $r_e \le C_k^*$ because of~\eqref{const:single_lower_bound}.
	
	Now, we transform the graph $G_I$ into a hypergraph $G$ as described in Section~\ref{sec:edge_scheduling}. Then, we define an orientation according to Algorithm \ref{alg:orientation} and run Algorithm \ref{alg:edge_scheduling} on $G$.
	
	By Lemma \ref{lem:edge_sched_solution}, this algorithm schedules every edge within $r_e + \lambda D_e$ in polynomial time. Given the specific structure of the hypergraph $G$ and the definition of deadlines for the hyperedges, the resulting schedule induces a feasible solution for the Coflow instance $I$ by assigning every operation to the slot of the corresponding hyperedge.
	
	Let $C_k$ be the final completion time of job $k$ in this solution; let $e$ be the last edge in the schedule associated to $k$; and let $C_e$ be the time slot in which $e$ is scheduled. Then for all $k \in [n]$:
	\begin{equation*}
		C_k = C_e \le r_e + \lambda D_e = r_k + \lambda D_k.
	\end{equation*}
	Summing over all jobs $k$, we obtain
	\begin{align*}
		\sum_{k=1}^n w_k C_k 	&\le \sum_{k=1}^n w_k (r_k + \lambda D_k) \\
								&= \sum_{k=1}^n w_k (r_k + \lambda \cdot 2 C_k^*) \\
								&\le (2 \lambda + 1) \cdot \sum_{k=1}^n w_k C_k^* \\
								&\le (2 \lambda + 1) \cdot \opt(I),
	\end{align*}
	and if $r_k \equiv 0$, we have
	\begin{align*}
		\sum_{k=1}^n w_k C_k 	&\le \sum_{k=1}^n w_k (\lambda D_k) \\
								&= 2\lambda \cdot \sum_{k=1}^n w_k C_k^* \\
								&\le 2\lambda \cdot \opt(I).
	\end{align*}
	We conclude that our Algorithm solves \gls{acr:pcsp} within a factor of $2\lambda + 1$ of the optimal solution for general release times. In the case of zero release times the solution lies within a factor of $2\lambda$ of the optimum.
\end{proof}

\subsection{Hardness of PCS}
\label{subsec:hardness_reduction}

In this section, we prove that \gls{acr:pcsp} admits no constant approximation ratio unless P = NP. More specifically, we prove the following, stronger statement:
\begin{theorem}
	\label{thm:pcs_hardness}
	Let $\mathcal{P} = \bigcup \pjk $ be the set of all paths in an instance of \gls{acr:pcsp}. Unless $\mathrm{P} = \mathrm{NP}$, there exists no algorithm for \gls{acr:pcsp} with an approximation ratio smaller than ${|\mathcal{P}|}^{1-\eps}$ for all $\eps > 0$, even if there is only a single coflow with unit weight.
\end{theorem}

\begin{proof}
	We prove the hardness of \gls{acr:pcsp} via a reduction from Vertex Coloring.
	
	Let $G = (V,E)$ be a given graph in an instance of Vertex Coloring. We construct a hypergraph $H = (V_H,E_H)$ whose edges correspond to the vertices of $G$. For every edge $e \in E$, we add a node $v_e$ to $H$, such that only the two hyperedges corresponding to the endpoints of $e$ contain $v_e$. Put differently, $G$ is the line graph of $H$.
	
	Clearly, $H$ can be constructed in polynomial time. Additionally, $|V| = |E_H|$ and the edges of $H$ can be colored with $k$ colors if and only if the vertices of $G$ can be colored with $k$ colors.
	
	To turn $H$ into an instance of \gls{acr:pcsp}, we introduce a single coflow $C$ with weight 1 and redefine the hyperedges $j \in E_H$ as paths $P_j^{(C)}$ of unit-sized flows $f_j^{(C)}$ on the instance graph $G_I$ constructed from $H$. Then the completion time of $C$ is equal to the number of colors used in a coloring of the edges of $H$.
	
	If $\mathrm{P} \neq \mathrm{NP}$, Vertex Coloring cannot be approximated within ${|V|}^{1-\eps}$ \cite{Zuckerman2006}. Therefore, under the same assumption, \gls{acr:pcsp} cannot be approximated within ${|\mathcal{P}|}^{1-\eps}$, even when restricted to instances with a single, unit-weight coflow.
\end{proof}

\section{Extensions of the Algorithm}
\label{sec:extensions}

This section generalizes our result to additional application cases. First, we show in Section~\ref{subsec:vertexConstraints} how the algorithm can be extended for general vertex constraints. Then, we apply our algorithm to \gls{acr:bcsp} in Section~\ref{subsec:bipartiteCS}. 

\subsection{General Vertex Constraints}
\label{subsec:vertexConstraints}

In this section, we show how our algorithm can be generalized to $(i)$~homogeneous vertex capacities greater than one and $(ii)$~heterogeneous vertex capacities.

In the homogeneous case it is simple to transform the problem back to the unit capacity case. In the heterogeneous case, the approximation ratio depends on the maximum ratio between the average and lowest capacity of the vertices of a hyperedge as we will show in the remainder of this section.

Let $G = (V,E)$ be a hypergraph as constructed in Section~\ref{sec:edge_scheduling} and let $u(v) \in {\mathds{Z}}_{>0}$ be given for all $v \in V$. For every hyperedge $e \in E$ we introduce the notions of \emph{average capacity} ($\avg(e)$) and \emph{capacity disparity} ($\Delta(e)$): 
\begin{equation*}
\avg(e) := \frac{\sum_{v \in e} u(v)}{|e|}, \qquad \Delta(e) := \left\lceil \frac{\avg(e)}{\min_{v \in e} u(v)} \right\rceil.
\end{equation*}

To this end, we show Theorem~\ref{thm:result_det_path_based_extended}. We note that for $\lambda = 1$, where the ``hyperedges" only consist of single vertices, $\Delta(e) = 1$ holds for all $e$. Hence, we retain the ratios of $3$ and $2$ in this generalization of \gls{acr:cosp}. As soon as edges consist of at least two vertices, we must include the capacity disparity in the approximation ratio.

\begin{theorem}
\label{thm:result_det_path_based_extended}
	Let $\Delta = \max_{e \in E}\Delta(e)$. There exists a $(2\lambda \Delta + 1)$-approximation for Path-based Coflow Scheduling with arbitrary release times and a $(2\lambda \Delta)$-approximation for Path-based Coflow Scheduling with zero release times.
\end{theorem}

We note that if all vertex capacities are homogeneous, that is $v(u) \equiv \bar{u}$, the capacity disparity of all edges is equal to $1$. Thus, in this case we retain the ratios $(2\lambda + 1)$ and $2\lambda$ from the unit capacity case. Alternatively, the homogeneous problem can be transformed back to the unit capacity case by linearly scaling the time horizon by $\bar{u}$, i.e.\ $\bar{u}$ timesteps in the new schedule correspond to $1$ timestep in the original problem. This incurs no additional factors in the approximation ratio of the algorithm.

Now consider general capacities $u(i)$ for every machine $i \in [m]$. We modify constraints~\eqref{const:single_lower_bound} of the LP to $C_k \ge r_k + \frac{\Lik}{u(i)}$ for all $k$ and $i$. Changing constraints~\eqref{const:subset_lower_bound} analogously, we get the following LP:
\begin{align*}
	\min	&&\sum_{k = 1}^n w_k C_k 	&								&		& \\
	\st		&&C_k						&\ge r_k + \frac{\Lik}{u(i)}	&\quad	& \forall k \in [n], \forall i \in [m]\\
			&&\sum_{k\in S}	\Lik C_k	&\ge \frac{f_i(S)}{u(i)}		&		& \forall S \subset [n], \forall i \in [m].
\end{align*}
Let $(C_k^*)_{k \in [n]}$ be an optimal solution of this LP, ordered such that $C_1^* \le \hdots \le C_n^*$. Then, Lemma~\ref{lem:bound_on_ext_LP_sol} revisits Lemma \ref{lem:bound_on_LP_sol}, requiring only minor changes in its proof.

\begin{lemma}
\label{lem:bound_on_ext_LP_sol}
	For all jobs $k \in [n]$ and all machines $i \in [m]$: $C_k^* \ge \frac{1}{2 u(i)} \sum_{l=1}^k \Lil$.
\end{lemma}

We again define $D_k = 2 \cdot C_k^*$ and consider the hypergraph $G = (V,E)$ constructed from the input graph $G_I$ where all operations correspond to hyperedges. We define release times and edge deadlines analogously to Section~\ref{sec:finding_deadlines}, but based on the updated LP.
Then, we use Algorithm \ref{alg:orientation} to construct an orientation $O$ of the line graph $L = L(G)$ and reformulate Lemma \ref{lem:orient_properties}. 

\begin{lemma}
\label{lem:orient_properties_extended}
	The orientation $O$ as constructed by Algorithm \ref{alg:orientation} in the case of general vertex capacities has the following properties:
	\begin{enumerate}
		\item Any vertex $e \in V(L)$ of the line graph satisfies $|d^+(e)| \le \lambda (D_e \cdot \avg(e) -1)$.
		\item It is kernel-perfect, i.e., every induced subgraph of $O$ has a kernel.
	\end{enumerate}
\end{lemma}

\begin{proof}
	We prove Lemma \ref{lem:orient_properties_extended} in two steps.
	\begin{enumerate}
		\item Let $e$ be any vertex of $L$ and $v$ be an endpoint of $e$. We may repeat the line of argument of the proof of Lemma \ref{lem:orient_properties} until the step
		\begin{equation*}
			\left|d_v^+(e)\right| \le \sum_{l=1}^{k_e} L_v^{(l)} - 1.
		\end{equation*}
		By Lemma \ref{lem:bound_on_ext_LP_sol} and the definition of $D_{k_e}$, we have $|d_v^+(e)| \le D_{k_e} \cdot u(v) - 1$. Note here that endpoints of $e$ correspond to machines of the job $k_e$. We sum over all such endpoints of $e$ to receive
		\begin{align*}
			\left|d^+(e)\right| &\le \sum_{v \in e} \left|d_v^+(e)\right| \\
								&\le \sum_{v \in e} (D_{k_e} \cdot u(v) - 1) \\
								&=  D_e \cdot \sum_{v \in e} u(v) - |e| \\
								&= |e| \cdot (D_e \cdot \avg(e) - 1).
		\end{align*}
		Observing that the number of endpoints $|e|$ is bounded by $\lambda$ gives the final inequality.
		
		\item See Proof of Lemma \ref{lem:orient_properties}.
	\end{enumerate}
\end{proof}

Now, we change the Edge Scheduling algorithm to nontrivial vertex constraints as follows.
\begin{algorithm}[h]
	\Input{A hypergraph $G=(V,E)$, an orientation $O$ of $L(G)$, a release time $r_e$ and a deadline $D_e$ for every $e \in E\ \st\ d^+(e) \le \lambda (D_e \cdot \avg(e) - 1)$, a set of vertex constraints $u(v)$ for all $v \in V$.}
	\Output{A feasible edge schedule on $G$.}
	$T \leftarrow \max_{e \in E} (r_e) + \lambda \cdot \max_{e \in E} (D_e \Delta(e))$\;
	\For{$t \leftarrow 1$ \KwTo $T$}{
		$O' \leftarrow$ 	induced subgraph of $O$ on all vertices $e$ with $t > r_e$\;
		\ForEach{$v \in V(G)$}{
			$c(v) \leftarrow u(v)$ \tcp*[r]{set capacity constraints}
		}
		\While{$O'\neq \emptyset$}{
					$U \leftarrow$ a kernel of $O'$\;
					\ForEach{$e \in U$}{
						schedule $e$ in timeslot $t$\;		
						\ForEach{$v \in V(G)$ incident to $e$}{$c(v) \leftarrow c(v)-1$\;}
					}
				remove $U$ from $O$ and $O'$\;
				\ForEach{$v \in V(G)$}{
					\lIf{$c(v)=0$}{remove all edges incident to $v$ from $O'$}
				}
		}
	}
	\caption{Edge Scheduling with general vertex constraints}\label{alg:vertex_constr_edge_scheduling}
\end{algorithm}

\begin{lemma}
	Algorithm \ref{alg:vertex_constr_edge_scheduling} finds a feasible solution for Edge Scheduling on a given hypergraph $G$, s.t.\ every hyperedge $e$ is scheduled not later than $r_e + \lambda D_e \Delta(e)$.
\end{lemma}

\begin{proof}
	We note that existance and construction of a kernel is equivalent to the proof of Lemma \ref{lem:edge_sched_solution}.
	
	Let $e$ be any hyperedge of $G$. We prove that after at most $r_e + \lambda D_e \Delta(e)$ time steps it holds that $d^+(e) = 0$ and that there is at least one open slot left for $e$ itself.
	
	By Lemma \ref{lem:orient_properties_extended}, the orientation $O$ fulfills $|d^+(e)| \le \lambda (D_e \cdot \avg(e) - 1)$ in the beginning of Algorithm~\ref{alg:vertex_constr_edge_scheduling}. Edge $e$ is in $O'$ in every iteration $r_e +1, r_e +2, \dots, T$. In every such iteration, we repeatedly search for a kernel of $O'$ until all vertices have no capacities left. One particular edge $e$ remains in $O'$ as long as all its endpoints have available capacity. Accordingly, unless it is already scheduled, $e$ is considered at least $\min_{v \in e} u(v)$ times in every slot $t$.

	If $e \in U$ at any point until iteration $r_e + \lambda D_e \Delta(e)$, then $e$ is scheduled and the claim holds. If, on the other hand, $e \notin U$ for all sub-iterations before that, then $e$ must have an outgoing edge to some $e' \in U$ in every such sub-iteration by the kernel property of $U$. Therefore, $e$ loses at least $\min_{v \in e} u(v)$ outgoing edges in every iteration.
	
	In total, $e$ would lose at least 
	\begin{equation*}
		\lambda D_e \Delta(e) \cdot \min_{v \in e} u(v) \ge \lambda D_e \cdot \avg(e)
	\end{equation*}
	outgoing edges until iteration $r_e + \lambda D_e \Delta(e)$. But since $e$ only has
	\begin{equation*}
		\left|d^+(e)\right| \le \lambda (D_e \cdot \avg(e) - 1) < \lambda D_e \cdot \avg(e)
	\end{equation*}
	such outgoing edges, there is at least one slot left where it holds that $|d^+(e)| = 0$. Hence $e$ is scheduled not later than iteration $r_e + \lambda D_e \Delta(e)$.
\end{proof}

To finally prove Theorem \ref{thm:result_det_path_based_extended}, we follow along the lines of the proof of Theorem~\ref{thm:result_det_path_based}. We estimate the completion time of a job $k$ by its latest edge $e$. Hence,
\begin{equation*}
	C_k = C_e \le r_e + \lambda D_e \Delta(e) \le r_k + \lambda D_k \cdot \max_{e \in E} \Delta(e).
\end{equation*}
For the final estimation we then get
\begin{align*}
	\sum_{k=1}^n w_k C_k 	&\le \sum_{k=1}^n w_k (r_k + \lambda D_k \cdot \max_{e \in E} \Delta(e)) \\
							&\le (2 \lambda \Delta + 1) \cdot \opt(I)
\end{align*}
and note that the case without release times is analogous.

\subsection{Bipartite Coflow Scheduling}
\label{subsec:bipartiteCS}

We now show how our algorithm can be applied to \gls{acr:bcsp}. An instance of \gls{acr:bcsp} considers a bipartite graph $G_I$, each side consisting of $m$ ports. Each coflow $k$ sends $\cijk$ units from input port $i$ to output port $j$. The definitions of weight, release time, and completion time are the same as in Section~\ref{subsec:def_path_based_coflow_scheduling}; each port can handle at most one unit-sized packet of data per time slot; and the objective remains to minimize $\sum_{k=1}^n w_k C_k$.

We define the load of job $k$ on machine $i$ as the sum of all operations on that machine. The load on machine $j$ is defined equivalently:
\begin{equation*}
	\Lik = \sum_{j=1}^n \cijk, \qquad \Ljk = \sum_{i=1}^n \cijk.
\end{equation*}

With this notation, we redefine LP~\eqref{eq:LP_obj} as
\begin{align*}
	\min	&&\sum_{k = 1}^n w_k C_k 	&				&		& \\
	\st		&&C_k						&\ge r_k + \Lik	&\quad	& \forall k \in [n], \forall i \in [m]  \\
			&&C_k						&\ge r_k + \Ljk	&\quad	& \forall k \in [n], \forall j \in [m]  \\
			&&\sum_{k\in S}	\Lik C_k	&\ge f_i(S)		&		& \forall S \subset [n], \forall i \in [m] \\
			&&\sum_{k\in S}	\Ljk C_k	&\ge f_j(S)		&		& \forall S \subset [n], \forall j \in [m] .
\end{align*}

Again, we consider an optimal solution $(C_k^*)_{k \in [n]}$ of the LP which is ordered such that ${C_1^* \le \hdots \le C_n^*}$. Then, Lemma \ref{lem:bound_on_LP_sol} holds without changes and we can analogously define $D_k = 2 \cdot C_k^* \ge \sum_{l=1}^k \Lil$.

In the bipartite case, every operation already corresponds directly to an edge in the graph such that transforming $G_I$ becomes superfluous. Analogously to our general case, we define the release times and deadlines of the edges based on the job the edge belongs to. The orientation is defined as in Algorithm \ref{alg:orientation} and Lemmas \ref{lem:orient_properties} and \ref{lem:edge_sched_solution} hold with $\lambda = 2$; the proofs are analogous. Moreover, the proof of Theorem \ref{thm:result_det_path_based} with $\lambda = 2$ is equivalent so that we can state our result for the bipartite case.
\begin{theorem}
	The \gls{acr:pcsp} algorithm can be applied to \gls{acr:bcsp} and gives a $5$-ap\-pro\-xi\-ma\-tion for arbitrary release times and a $4$-approximation for zero release times.
\end{theorem}

In this context, we clarify that the algorithm of Ahmadi et al.~\cite{AhmadiKhullerEtAl2017} is not applicable to our more general \gls{acr:pcsp}: We recall that they based their approach on a primal-dual analysis of the LP relaxation to receive an order for the jobs. The main idea of their algorithm is a combinatorial shifting of operations from later to earlier time slots based on this job order. They use a result by \cite{QiuSteinEtAl2015} to schedule single jobs within a number of steps equal to their maximum load. 

We now prove that this central lemma does not generalize, even in the case $\lambda = 2$ if the graph is non-bipartite. With our notation this lemma is as follows.
\begin{lemma}[{\cite[Lemma 1]{AhmadiKhullerEtAl2017}}]
\label{lem:single_coflow_schedule}
	There exists a polynomial algorithm that schedules a single coflow job $k$ in $\Lk := \max_{i \in [m]} \Lik$ time steps.
\end{lemma}

\begin{figure}[ht]
	\centering
	\begin{tikzpicture}[main node/.style={draw,circle,inner sep=2pt,fill,thick},>=latex]
	\node[main node] (1) at (0,0) {};
	\node[main node] (2) at (2,0) {};
	\node[main node] (3) at (1,2) {};
	\draw (1) -- (2);
	\draw (1) -- (3);
	\draw (2) -- (3);
	\draw[ultra thick,green!50!black,->] (1) -- node[above] {1} 		(2);
	\draw[ultra thick,green!50!black,<-] (2) -- node[above right] {1} 	(3);
	\draw[ultra thick,green!50!black,->] (3) -- node[above left] {1} 	(1);
	\end{tikzpicture}
	\caption{Example graph with $\lambda = 2$ for the case of non-bipartite graphs}
	\label{pic:counterexample_generalization}
\end{figure}
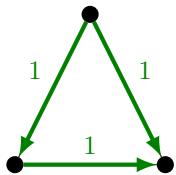

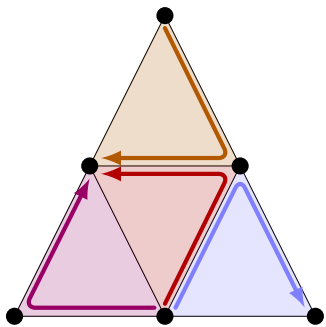
\begin{figure}[ht]
	\centering
		\begin{tikzpicture}[main node/.style={draw,circle,inner sep=2pt,fill,thick},>=latex]
		\coordinate (B2) at (1,0) {};
		\coordinate (A2) at (3,4) {};
		\coordinate (C1) at (2,2) {};
		\coordinate (A1) at (3,0) {};
		\coordinate (B1) at (4,2) {};
		\coordinate (C2) at (5,0) {};
		\coordinate (B2) at (1,0) {};
		\coordinate (A2) at (3,4) {};
		\coordinate (C1) at (2,2) {};
		\coordinate (A1) at (3,0) {};
		\coordinate (B1) at (4,2) {};
		\coordinate (C2) at (5,0) {};
		\draw (A1) -- (B1);
		\draw (A1) -- (C1);
		\draw (A1) -- (B2);
		\draw (A1) -- (C2);
		\draw (B1) -- (A2);
		\draw (B1) -- (C1);
		\draw (B1) -- (C2);
		\draw (C1) -- (B2);
		\draw (C1) -- (A2);
		\fill[color=blue!50!white,opacity=.2] (A1) -- (B1) -- (C2) -- cycle;
		\fill[color=red!60!blue,opacity=.2] (A1) -- (B2) -- (C1) -- cycle;
		\fill[color=red!70!black,opacity=.2] (A1) -- (B1) -- (C1) -- cycle;
		\fill[color=orange!70!black,opacity=.2] (A2) -- (B1) -- (C1) --cycle;
		\node[main node] (NA1) at (A1) {};
		\node[main node] (NB1) at (B1) {};
		\node[main node] (NC1) at (C1) {};
		\node[main node] (NA2) at (A2) {};
		\node[main node] (NB2) at (B2) {};
		\node[main node] (NC2) at (C2) {};
		\draw[cap=round,rounded corners,ultra thick,->,blue!50!white] ([xshift=4pt,yshift=0pt]NA1.north) --  ([xshift=0pt,yshift=-1.5pt]NB1.south) -- ([xshift =-4pt,yshift=0pt]NC2.north);
		\draw[cap=round,rounded corners,ultra thick,->,red!60!blue] ([xshift=-4pt,yshift=0pt]NA1.north) -- ([xshift=4pt,yshift=0pt]NB2.north) -- ([yshift=-1.5pt]NC1.south);
		\draw[cap=round,rounded corners,ultra thick,->,red!70!black] ([yshift=1.5pt]NA1.north) -- ([xshift=-1pt,yshift=-3pt]NB1.west) -- ([xshift=1pt,yshift=-3pt]NC1.east);
		\draw[cap=round,rounded corners,ultra thick,->,orange!70!black] ([xshift=0pt,yshift=-1.5pt]NA2.south) -- ([xshift=-1pt,yshift=3pt]NB1.west) -- ([xshift=1pt,yshift=3pt]NC1.east);
		\end{tikzpicture}
	\caption{Example graph with $\lambda = 3$ for the case of $\lambda$-partite hypergraphs}
	\label{pic:counterexample_k-partite}
\end{figure}

We consider a simple example graph (see Figure~\ref{pic:counterexample_generalization}) consisting of three vertices connected by three edges which form a triangle. The single coflow $k$ on this graph is defined by an operation $f_e^{(k)}$ on each edge $e$ with $c_e^{(k)} = 1$. The load on any vertex of the graph is equal to two, hence $\Lk = 2$. However, since the edges form a triangle, three steps are needed to feasibly schedule the entire job and Lemma \ref{lem:single_coflow_schedule} does not hold.

Additionally, the lemma does not hold when we generalize the bipartiteness condition to $\lambda$-partite hypergraphs, which arise in \gls{acr:pcsp} with $\lambda>2$. Figure~\ref{pic:counterexample_k-partite} shows a counterexample that consists of one coflow $k$ with four flows. Each flow sends one unit of data along a path with three vertices ($\lambda=3$). The corresponding hypergraph is $3$-partite with the start vertices, middle vertices, and end vertices of the flow paths forming the three disjoint vertex sets. Flows and their corresponding hyperedges have the same color. Because any two hyperedges have a common vertex, any feasible schedule requires at least four time steps. This contradicts Lemma \ref{lem:single_coflow_schedule} as the maximum load is only $\Lk=3$.

\section{Conclusion}

In this paper, we introduced the \glsentrylong{acr:pcsp} problem with release dates that arises in the context of today's distributed computing projects. We presented a $(2\lambda +1)$-approximation algorithm for homogeneous unit-sized node capacities. For zero release times this result improves to a $(2\lambda)$-ap\-pro\-xi\-ma\-tion. We validated our results by showing that \gls{acr:pcsp} admits no constant approximation ratio unless P = NP. We generalized our algorithm to arbitrary node constraints with a $(2\lambda \Delta +1)$- and a $(2\lambda \Delta)$-approximation in the case of general and zero release times. Here, $\Delta$ captures the capacity disparity between nodes. Furthermore, we showed that our algorithm is applicable to a wide range of problem variants, often matching the state of the art, e.g., for \glsentrylong{acr:bcsp} and \glsentrylong{acr:cosp}.

Further work is required in closing the gaps between the presented ratios and the lower bound of $2$ given by the reduction to \gls{acr:cosp}, which is not tight for fixed $\lambda \ge 2$. It is likely that our robust approach using orientations to sort the operations in the scheduling part of our algorithm can be further improved with new ideas.

Finally, we remark that it might be possible to modify and extend the approach of \cite{AhmadiKhullerEtAl2017} to our general framework, since the given counterexamples only contradict Lemma \ref{lem:single_coflow_schedule} but do not yield a worse approximation ratio. We also leave this question open for future research to deliberate.

\newpage
\bibliographystyle{abbrv}
\bibliography{Path_based_Coflow_Scheduling}

\begin{thebibliography}{10}

\bibitem{AhmadiKhullerEtAl2017}
S.~Ahmadi, S.~Khuller, M.~Poruhit, and S.~Yang.
\newblock On scheduling coflows.
\newblock In {\em Integer Programming and Combinatorial Optimization}, pages
  13--24. Springer International Publishing, 2017.

\bibitem{ChowdhuryKhullerEtAl2019}
M.~Chowdhury, S.~Khuller, M.~Purohit, S.~Yang, and J.~You.
\newblock Near optimal coflow scheduling in networks.
\newblock In {\em 31st ACM Symposium on Parallelism in Algorithms and
  Architectures}, SPAA '19, 2019.

\bibitem{ChowdhuryStoica2012}
M.~Chowdhury and I.~Stoica.
\newblock Coflow: A networking abstraction for cluster applications.
\newblock In {\em Proceedings of the 11th ACM Workshop on Hot Topics in
  Networks}, HotNets-XI, pages 31--36. ACM, 2012.

\bibitem{ChowdhuryStoica2015}
M.~Chowdhury and I.~Stoica.
\newblock Efficient coflow scheduling without prior knowledge.
\newblock In {\em Proceedings of the 2015 ACM Conference on Special Interest
  Group on Data Communication}, SIGCOMM '15, pages 393--406, New York, NY, USA,
  2015. ACM.

\bibitem{ChowdhuryZhongEtAl2014}
M.~Chowdhury, Y.~Zhong, and I.~Stoica.
\newblock Efficient coflow scheduling with varys.
\newblock In {\em Proceedings of the 2014 ACM Conference on SIGCOMM}, SIGCOMM
  '14, pages 443--454, New York, NY, USA, 2014. ACM.

\bibitem{DeanGhemawat2008}
J.~Dean and S.~Ghemawat.
\newblock Mapreduce: Simplified data processing on large clusters.
\newblock {\em Commun. ACM}, 51(1):107--113, Jan. 2008.

\bibitem{GargKumarEtAl2007}
N.~Garg, A.~Kumar, and V.~Pandit.
\newblock Order scheduling models: Hardness and algorithms.
\newblock In V.~Arvind and S.~Prasad, editors, {\em FSTTCS 2007: Foundations of
  Software Technology and Theoretical Computer Science}, pages 96--107, Berlin,
  Heidelberg, 2007. Springer Berlin Heidelberg.

\bibitem{GoogleDataflow2019}
Google.
\newblock Google cloud dataflow.
\newblock https://cloud.google.com/dataflow/, 2019.

\bibitem{GroetschelLovaszEtAl1993}
M.~Gr\"otschel, L.~Lovász, and A.~Schrijver.
\newblock {\em Geometric Algorithms and Combinatorial Optimization}.
\newblock Springer Berlin Heidelberg, Berlin, Heidelberg, 1993.

\bibitem{ImMoseley2019}
S.~Im, B.~Moseley, K.~Pruhs, and M.~Purohit.
\newblock {Matroid Coflow Scheduling}.
\newblock In C.~Baier, I.~Chatzigiannakis, P.~Flocchini, and S.~Leonardi,
  editors, {\em 46th International Colloquium on Automata, Languages, and
  Programming (ICALP 2019)}, volume 132 of {\em Leibniz International
  Proceedings in Informatics (LIPIcs)}, pages 145:1--145:14, Dagstuhl, Germany,
  2019. Schloss Dagstuhl--Leibniz-Zentrum fuer Informatik.

\bibitem{ImPurohit2017}
S.~Im and M.~Purohit.
\newblock A tight approximation for co-flow scheduling for minimizing total
  weighted completion time.
\newblock {\em CoRR}, abs/1707.04331, 2017.

\bibitem{JahanjouKantorEtAl2017}
H.~Jahanjou, E.~Kantor, and R.~Rajaraman.
\newblock Asymptotically optimal approximation algorithms for coflow
  scheduling.
\newblock In {\em Proceedings of the 29th ACM Symposium on Parallelism in
  Algorithms and Architectures}, SPAA '17, pages 45--54. ACM, 2017.

\bibitem{LeungLiEtAl2007}
J.~Y.-T. Leung, H.~Li, and M.~Pinedo.
\newblock Scheduling orders for multiple product types to minimize total
  weighted completion time.
\newblock {\em Discrete Applied Mathematics}, 155(8):945 -- 970, 2007.

\bibitem{MastrolilliQueyranneEtAl2010}
M.~Mastrolilli, M.~Queyranne, A.~S. Schulz, O.~Svensson, and N.~A. Uhan.
\newblock Minimizing the sum of weighted completion times in a concurrent open
  shop.
\newblock {\em Operations Research Letters}, 38(5):390 -- 395, 2010.

\bibitem{QiuSteinEtAl2015}
Z.~Qiu, C.~Stein, and Y.~Zhong.
\newblock Minimizing the total weighted completion time of coflows in
  datacenter networks.
\newblock In {\em Proceedings of the 27th ACM Symposium on Parallelism in
  Algorithms and Architectures}, SPAA '15, pages 294--303, New York, NY, USA,
  2015. ACM.

\bibitem{Queyranne1993}
M.~Queyranne.
\newblock Structure of a simple scheduling polyhedron.
\newblock {\em Mathematical Programming}, 58(1):263--285, Jan 1993.

\bibitem{Richardson1946}
M.~Richardson.
\newblock On weakly ordered systems.
\newblock {\em Bull. Amer. Math. Soc.}, 52(2):113--116, 02 1946.

\bibitem{SachdevaSaket2013}
S.~{Sachdeva} and R.~{Saket}.
\newblock Optimal inapproximability for scheduling problems via structural
  hardness for hypergraph vertex cover.
\newblock In {\em 2013 IEEE Conference on Computational Complexity}, pages
  219--229, 2013.

\bibitem{ShafieeGhaderi2018}
M.~Shafiee and J.~Ghaderi.
\newblock An improved bound for minimizing the total weighted completion time
  of coflows in datacenters.
\newblock {\em IEEE/ACM Trans. Netw.}, 26(4):1674--1687, 2018.

\bibitem{ShiZhangEtAl2018}
L.~Shi, J.~Zhang, Y.~Liu, and T.~G. Robertazzi.
\newblock Coflow scheduling in data centers: Routing and bandwidth allocation.
\newblock {\em CoRR}, abs/1812.06898, 2018.

\bibitem{WangCheng2007}
G.~Wang and T.~E. Cheng.
\newblock Customer order scheduling to minimize total weighted completion time.
\newblock {\em Omega}, 35(5):623 -- 626, 2007.

\bibitem{ZahariaChowdhuryEtAl2010}
M.~Zaharia, M.~Chowdhury, M.~J. Franklin, S.~Shenker, and I.~Stoica.
\newblock Spark: Cluster computing with working sets.
\newblock {\em HotCloud}, 10(10-10):95, 2010.

\bibitem{ZhaoChenEtAl2015}
Y.~{Zhao}, K.~{Chen}, W.~{Bai}, M.~{Yu}, C.~{Tian}, Y.~{Geng}, Y.~{Zhang},
  D.~{Li}, and S.~{Wang}.
\newblock Rapier: Integrating routing and scheduling for coflow-aware data
  center networks.
\newblock In {\em 2015 IEEE Conference on Computer Communications (INFOCOM)},
  pages 424--432, April 2015.

\bibitem{Zuckerman2006}
D.~Zuckerman.
\newblock Linear degree extractors and the inapproximability of max clique and
  chromatic number.
\newblock In {\em Proceedings of the Thirty-Eighth Annual ACM Symposium on
  Theory of Computing}, STOC ’06, page 681–690, New York, NY, USA, 2006.
  Association for Computing Machinery.

\end{thebibliography}

\newpage
\appendix

\section{Equivalence of Edge Constraints and Node Constraints}
\label{sec:equ_edge_node_constr}

In this section we show how to simulate edge capacities as used in \cite{ChowdhuryKhullerEtAl2019,JahanjouKantorEtAl2017} by node capacities and vice versa. 

For the forward direction, we assume a graph $G$ with capacities on every edge as well as a set of coflows defined by operations on given paths. The reduction works as follows: We split every edge $e$ in the middle and add a node $v_e$ with the corresponding capacity of the split edge. All other nodes are assigned infinite capacity. Evidently, this transformation is polynomial and correct. 

In Section \ref{subsec:contribution} we argue that our algorithm can solve circuit-based coflows with unit-capacities with an approximation ratio smaller than $17.6$ for all $\lambda \le 9$. However, the above reduction increases the lengths of the given paths in our problem, in particular it increases the parameter $\lambda$, on which the approximation factor of our algorithm depends. In the following, we show that the approximation ratio is still less than $17.6$ if $\lambda \le 9$, based on the value of $\lambda$ in the original graph $G$.

Let the set of machines be partitioned into two sets of nodes, $[m] = I_{<\infty} \dot{\cup} I_{\infty}$, where $I_{<\infty}$ is the set of machines with finite capacity and $I_{\infty}$ is the set of machines with infinite capacity. By looking at the LP relaxation \eqref{eq:LP_obj}, we see that the constraints for a node $i \in I_{\infty}$ do not need to be added to the LP since they do not limit the completion times of the operations. In fact, if all nodes had infinite capacity, one constraint of the form $C_k \ge r_k + 1$ for all $k$ would suffice to describe the polyhedron completely.

Therefore, we only need to consider the nodes in $I_{<\infty}$ for the definition of deadlines $D_k$. Additionally, we can simplify the construction of the line graph in Section \ref{sec:edge_scheduling} such that an edge in $L$ between two vertices $e$ and $f$ is only added if $e$ and $f$ share a node with finite capacity.

We redefine $\lambda_{<\infty}$ as the maximum number of \emph{finite} nodes in a longest flow-bearing path in the graph. Then Lemmas \ref{lem:orient_properties} and \ref{lem:edge_sched_solution} hold with this new definition, since the deadlines $D_k$ were defined using only such finite nodes. Finally, Theorem \ref{thm:result_det_path_based} can be amended such that there exists a $(2\lambda_{<\infty} + 1)$-approximation and $2\lambda_{<\infty}$-approximation, respectively.

For our reduction described above, we see that for every path in the original graph $G$, the number of finite nodes in the reduced setting is exactly equal to the number of edges in the path. Hence, for a given problem instance with parameter $\lambda$, our algorithm gives us a $(2(\lambda-1) + 1)$-approximation and $2(\lambda-1)$-approximation, respectively. Therefore, if $\lambda \le 9$, the ratio of our algorithm is smaller than $17.6$.

To show that node capacities can be simulated via edge capacities, we refer to \cite{ChowdhuryKhullerEtAl2019}. There, it is stated that this can be done by replacing every node by a gadget consisting of two nodes and setting the capacity of the new edge as the capacity of the old node.

\section{Constructing a Kernel in a Cycle-free Directed Graph}
\label{sec:constructing_kernel}

We present a simple algorithm on how to find a kernel $U$ in a directed graph $O = (V,E)$ without directed cycles in Algorithm \ref{alg:find_kernel}.

\begin{algorithm}[htb]
	$U \leftarrow \emptyset$\;
	\While{$O$ is non-empty}{
		$U_0 \leftarrow \{v \in V: v$ has out-degree 0$\}$\;
		$U \leftarrow U \cup U_0$\;
		let $N(U_0)$ be the nodes with an arc towards some node in $U_0$\;
		remove $U_0$, $N(U_0)$ and all adjacent arcs from $O$\;
	}
	\Return U;
	\caption{Finding a kernel in a directed graph without cycles} \label{alg:find_kernel}
\end{algorithm}

The runtime of this algorithm is clearly polynomial, since all nodes are considered at most once. Nodes with out-degree 0 can be found in linear time.

For correctness, we first verify the termination of the algorithm. If $O$ is non-empty, then the set $U_0$ must contain at least one node. Assume $U_0 = \emptyset$, then every node has at least one outgoing arc which implies the existence of a directed cycle. Hence, at least one node gets removed from $O$ in every iteration, until $O$ is empty.

Now consider the properties of a kernel $U$ from Definition \ref{def:kernel}: it is an independent set and for all $v \notin U$ there exists an arc from $v$ to some node in $U$. The set $U$ as constructed by the algorithm above is clearly independent, since nodes of out-degree 0 cannot be adjacent themselves and all other adjacent nodes of $U_0$ are removed from $O$ in every iteration.

For the second property, consider $v \notin U$. Then, $v$ was in $N(U_0)$ for some iteration $i$ of the algorithm by the termination property. By definition, $v$ has an outgoing arc towards $U_0 \subset U$ in iteration $i$. Therefore, the above algorithm correctly returns a kernel of $O$ in polynomial time.

\section{Minor Improvement of Approximation Ratio}
\label{sec:improvements}

The inequalities used in the proofs of Lemma~\ref{lem:bound_on_LP_sol}, Lemma~\ref{lem:edge_sched_solution}, and Theorem~\ref{thm:result_det_path_based} leave some room for minor improvements of the approximation guarantees. In the following, we show how to modify these proofs in order to obtain a slightly improved version of Theorem~\ref{thm:result_det_path_based}:

\begin{theorem}
\label{thm:result_det_path_based_revisited}
Let $n$ be the number of jobs. There exists a $\left(\frac{2n}{n+1}\lambda + 1\right)$-approximation for \gls{acr:pcsp} with arbitrary release times and a $\left(\frac{2n}{n+1}\lambda\right)$-approximation for \gls{acr:pcsp} when all release times are smaller than $\lambda$.
\end{theorem}

We start by revisiting Lemma~\ref{lem:bound_on_LP_sol} and proving a slightly more precise version.

\begin{lemma}
\label{lem:bound_on_LP_sol_revisited}
	For all jobs $k \in [n]$ and all machines $i \in [m]$ the following holds:
	\begin{equation*}
		C_k^* \ge \frac{k+1}{2k} \cdot \sum_{l=1}^k \Lil.
	\end{equation*}
\end{lemma}

\begin{proof}
	Let $k \in [n], i \in [m], S = \{1,\dots,k\}$. Since $C_1^*,\dots,C_n^*$ is a feasible solution of the LP, it must fulfill constraint \eqref{const:subset_lower_bound} of the LP for $S$ and $i$:
	\[	\sum_{l=1}^k \Lil C_l^* = \sum_{l\in S} \Lil C_l^*\ge f_i(S)= \frac{1}{2} \cdot \left(\sum_{l=1}^k \left(\Lil\right)^2 + \left(\sum_{l=1}^k \Lil \right)^2 \right). \]
	Applying the Cauchy-Schwarz inequality on the first summand yields
	\[k \cdot \sum_{l=1}^k \left(\Lil\right)^2 = \sum_{l=1}^k 1^2 \cdot \sum_{l=1}^k \left(\Lil\right)^2  \ge \left(\sum_{l=1}^k \Lil \right)^2. \]
	Thus, we obtain the following lower bound for the completion time of job $k$:
	\begin{align*}
		C_k^* 	&\ge \frac{\sum_{l=1}^k \Lil C_k^*}{\sum_{l=1}^k \Lil} \\
						&\ge \frac{\sum_{l=1}^k \Lil C_l^*}{\sum_{l=1}^k \Lil} \\
						&\ge \frac{\left(\sum_{l=1}^k \left(\Lil\right)^2 + \left(\sum_{l=1}^k \Lil \right)^2 \right)}{2\cdot\sum_{l=1}^k \Lil} \\
						&\ge \frac{\left(\frac{1}{k} \cdot\left(\sum_{l=1}^k \Lil \right)^2 + \left(\sum_{l=1}^k \Lil \right)^2 \right)}{2\cdot\sum_{l=1}^k \Lil} \\
						&\ge \frac{(k+1)\cdot \left(\sum_{l=1}^k \Lil \right)^2}{2k\cdot\sum_{l=1}^k \Lil} \\
						&\ge \frac{k+1}{2k} \cdot\sum_{l=1}^k \Lil.
	\end{align*}
\end{proof}

This allows us to set the deadline for job $k$ to $D_k:=\frac{2k}{k+1} \cdot C_k^* $ instead of $2 \cdot C_k^*$ in equation~(\ref{eq:lower_bound_deadlines}).  These new deadlines later effect the approximation factor to be strictly smaller than $(2\lambda + 1)$, and $2\lambda$, respectively.

For the second improvement, which is the extension of the special case from zero release times to release times smaller than $\lambda$, we focus on a detail in the proof of Lemma~\ref{lem:edge_sched_solution}. In the end of the proof, we stated that after at most $\lambda(D_e-1)\le \lambda D_e-1$ iterations, we have $|d^+(e)|=0$. According to the first part of this inequality, we can restate Lemma~\ref{lem:edge_sched_solution} with a slightly stricter conclusion.

\begin{lemma}
\label{lem:edge_sched_solution_revisited}
Algorithm \ref{alg:edge_scheduling} finds a feasible solution for Edge Scheduling on a given hypergraph $G$, s.t.\ every hyperedge $e$ is scheduled not later than $r_e + \lambda D_e-(\lambda-1)$.
\end{lemma}

For the proof of Theorem~\ref{thm:result_det_path_based_revisited}, we proceed analogously to the proof of Theorem~\ref{thm:result_det_path_based}. Hence, we only need to modify the final computations for the approximation guarantee. Applying Lemma~\ref{lem:edge_sched_solution_revisited}, we obtain that for the final completion time $C_k$ of job $k$ in the solution provided by our algorithm
\[C_k\le r_k+\lambda D_k -(\lambda-1)\]
holds for all $k\in[n]$. Summing over all jobs yields

\begin{align}
	\sum_{k=1}^n w_k C_k 	&\le \sum_{k=1}^n w_k \left(r_k+\lambda D_k -(\lambda-1)\right) \nonumber\\
													&= \sum_{k=1}^n w_k \left(r_k+\lambda \cdot \frac{2k}{k+1} \cdot C_k^* -(\lambda-1)\right) \nonumber\\
													&=  \sum_{k=1}^n w_k\cdot \lambda \cdot \frac{2k}{k+1} \cdot C_k^*  + \sum_{k=1}^n w_k \left(r_k-\lambda+1\right) \label{eq:calc_approx_guarantee_revisited}\\
													&\le \frac{2n}{n+1}\cdot \lambda \cdot \sum_{k=1}^n w_k\cdot C_k^* + \sum_{k=1}^n w_k\cdot C_k^* \nonumber\\
													&\le \left(\frac{2n}{n+1}\cdot\lambda+1\right) \cdot \opt(I).\nonumber
\end{align}
This proves the first part of Theorem~\ref{thm:result_det_path_based_revisited}. For the case that $r_k<\lambda$ for all jobs $k\in[n]$, we have $r_k-\lambda+1\leq 0$, and, hence, the second sum in~(\ref{eq:calc_approx_guarantee_revisited}) is non-positive. Consequently, we have
\begin{align*}
	\sum_{k=1}^n w_k C_k 	&\le \sum_{k=1}^n w_k\cdot \lambda \cdot \frac{2k}{k+1} \cdot C_k^*  + \sum_{k=1}^n w_k \left(r_k-\lambda+1\right)\\
													&\le \frac{2n}{n+1}\cdot \lambda \cdot \sum_{k=1}^n w_k\cdot C_k^*\\
													&\le \frac{2n}{n+1}\cdot\lambda \cdot \opt(I),
\end{align*}
which establishes the second statement of Theorem~\ref{thm:result_det_path_based_revisited}. Note that $\frac{2n}{n+1}<2$ for any $n\in\mathbb{N}$, which means Theorem~\ref{thm:result_det_path_based_revisited} slightly improves the result of Theorem~\ref{thm:result_det_path_based}.

\end{document}